\documentclass[conference]{IEEEtran}
 \input{sty/ieeetran/complementary_newdefines.cls}

\usepackage{amssymb}
\usepackage{amsfonts}

\usepackage{subfig}
\usepackage{wrapfig}
\usepackage{float}
\usepackage{multirow, graphicx}
\usepackage{algorithm}
\usepackage{algpseudocode} 
\usepackage{amsmath}

\usepackage{xspace}

\usepackage{url}

\usepackage{sty/widetext} 

\usepackage{stackrel}

\let\emptyset\varnothing

\ifdefined\TTSTYLETARGETreport
\usepackage[
    backend=bibtex,
    maxbibnames=100,
    firstinits=true,
    bibstyle=numeric,
    citestyle=numeric
]{biblatex}
\fi

\usepackage{pifont}
\newcommand{\ignore}[1]{}

\usepackage{graphicx, multirow}

\usepackage{tabularx, multirow, rotating} 
\usepackage{subfig} 
\usepackage{algorithm} 
\usepackage{algpseudocode} 
\usepackage{ulem} 
\usepackage{color} 

\usepackage{listings} 
\lstdefinestyle{Oracle}{basicstyle=\ttfamily,
                        keywordstyle=\lstuppercase,
                        emphstyle=\itshape,
                        showstringspaces=true,
                        }
\makeatletter
\newcommand{\lstuppercase}{\uppercase\expandafter{\expandafter\lst@token
                           \expandafter{\the\lst@token}}}
\newcommand{\lstlowercase}{\lowercase\expandafter{\expandafter\lst@token
                           \expandafter{\the\lst@token}}}
\makeatother

\usepackage{tikz}
\newcommand*\circled[1]{\tikz[baseline=(char.base)]{
            \node[shape=circle,draw,inner sep=1pt] (char) {#1};}} 

\newif\ifboldnumber

\algrenewcommand\alglinenumber[1]{%
  \footnotesize\ifboldnumber\bfseries\fi\global\boldnumberfalse#1:}

\newtheorem{theorem}{Theorem}[section]
\newtheorem{lemma}[theorem]{Lemma}

\newtheorem{definition}[theorem]{Definition}

\usepackage{makecell}

\UseRawInputEncoding
\ifCLASSINFOpdf
\else
\fi
\begin{document}


\title{Asymmetric Mempool DoS Security: Formal Definitions and Provable Secure Designs}










\author{
  Wanning Ding \\
  Department of Computer Science \\
  Syracuse University \\
  \texttt{wding04@syr.edu}
  \and
  Yibo Wang \\
  Department of Computer Science \\
  Syracuse University \\
  \texttt{ywang349@syr.edu}
  \and
  Yuzhe Tang \\
  Department of Computer Science \\
  Syracuse University \\
  \texttt{ytang100@syr.edu}
}



\maketitle

\begin{abstract}
The mempool plays a crucial role in blockchain systems as a buffer zone for pending transactions before they are executed and included in a block. However, existing works primarily focus on mitigating defenses against already identified real-world attacks. This paper introduces secure blockchain-mempool designs capable of defending against any form of asymmetric eviction DoS attacks. We establish formal security definitions for mempools under the eviction-based attack vector. Our proposed secure transaction admission algorithm, named \textsc{saferAd-CP}, ensures eviction-security by providing a provable lower bound on the cost of executing eviction DoS attacks. Through evaluation with real transaction trace replays, \textsc{saferAd-CP} demonstrates negligible latency and significantly high lower bounds against any eviction attack, highlighting its effectiveness and robustness in securing blockchain mempools.

\ignore{
This paper presents \textsc{mpfuzz}, the first blockchain fuzzer that automatically and systematically finds known and unknown mempool DoS bugs in Ethereum. Fuzzing mempools, unlike fuzzing other blockchain components tackled in the existing literature, poses new design problems. To solve them, \textsc{mpfuzz} explores the previously unexplored transaction space, namely invalid transactions with varying fees, and defines a bug oracle specifically for asymmetric mempool DoS. \textsc{mpfuzz} mutates transactions and explores mempool states, {\it symbolically}, that is, based on the transaction symbols tailored to Ethereum protocols. To search deep-state bugs, \textsc{mpfuzz} optimistically estimates mempool costs. 

Running \textsc{mpfuzz} leads to the discovery of new mempool DoS bugs in leading Ethereum clients of the latest versions, all of which have been patched against the known bugs. In particular, all Ethereum clients are found vulnerable to at least one of two DoS patterns:  mempool eviction or locking.


This paper also presents a defensive mempool framework, \textsc{saferAd}, that provides provable security against the newly discovered DoS patterns. \textsc{saferAd} upper-bounds the victim transaction fees under mempool locking and lower-bounds the attacker transaction fees under mempool evictions.
}




\end{abstract} 



\providecommand{\ssssp}{{\sc SS\_SSP}\xspace}
\newcommand{\tremark}[1]{\footnote{\textcolor{red}{(Ting's comment: #1)}}}
\newcommand{\xremark}[1]{\footnote{\textcolor{red}{(Xin's comment: #1)}}}
\newcommand{\jj}[1]{\footnote{\textcolor{blue}{(Jiyong: #1)}}}
\newcommand{\yz}[1]{\footnote{\textcolor{red}{(Yuzhe: #1)}}}

\definecolor{mygreen}{rgb}{0,0.6,0}
\section{Introduction}

In public blockchains, a mempool is a data structure residing on every blockchain node, responsible for buffering unconfirmed transactions before they are included in blocks. On Ethereum, mempools are present in various execution-layer clients serving public transactions, such as Geth~\cite{me:geth}, Nethermind~\cite{me:nethermind}, Erigon~\cite{me:erigon}, Besu~\cite{me:besu}, and Reth~\cite{me:reth}. They are also utilized by block builders managing private transactions within the proposer-builder separation architecture, such as Flashbots~\cite{me:flashbotbuilder}, Eigenphi~\cite{me:eigenphibuilder}, and BloXroute builders~\cite{me:bloXroutebuilder}.

Unlike conventional network stacks, mempools are permissionless and must accept transactions from unauthenticated accounts to maintain decentralization. This open nature, while essential for decentralization, makes the mempool vulnerable to denial-of-service (DoS) attacks. In such an attack, an adversary infiltrates the target blockchain network, establishes connections with victim nodes, and floods them with crafted transactions to deny mempool services to legitimate transactions. The disruption of mempool services can severely impact various blockchain subsystems, including block building, transaction propagation, blockchain value extraction (e.g., MEV searching), remote-procedure calls, and Gas stations. For instance, empirical studies have shown that disabling mempools can force the Ethereum network to produce empty blocks, undermining validators' incentives and increasing the risk of 51\% attacks.

\noindent{\bf Related works \& open problems}:
Denial of mempool services have been recently recognized and studied in both the research community and industry. In a large-scale blockchain, a mempool of limited capacity has to order transactions by certain priority criteria for transaction admission and eviction. The admission policy can be exploited to mount mempool DoS. The early attack designs~\cite{DBLP:conf/fc/BaqerHMW16,DBLP:conf/icbc2/SaadNKKNM19} work by sending spam transactions of high prices to evict benign transactions of normal prices. These attacks are extremely expensive and are not practical. 

Of more realistic threat is the {\it Asymmetric} DeniAl of Mempool Service, coined by us as ADAMS attacks, where the attacker spends much less fees in the adversarial transactions he sends than the damage he caused, that is, the fees of evicted benign transactions that would be otherwise chargeable. DETER~\cite{DBLP:conf/ccs/LiWT21} is the first ADAMS attack studied in research works where the attacker sends invalid and thus unchargeable transactions (e.g., future transactions in Ethereum) to evict valid benign transactions from the mempool. 
The MemPurge attack~\cite{cryptoeprint:2023/956} similarly evict Ethereum's (more specifically, Geth’s) pending-transaction mempool by crafting invalid overdraft transactions. 
These ADAMS attacks follow a single-step pattern and can be easily detected. In fact, DETER bugs have been fixed in Geth $v1.11.4$~\cite{me:gethfix11:thirdparty}, and there is a code fix implemented and tested against MemPurge on Geth~\cite{me:mempurge:fix}. 

More sophisticated and stealthy attacks are recently discovered by \textsc{mpfuzz}, a symbolized stateful fuzzer for testing mempools~\cite{wang2023understanding}. The discovered attacks transition mempool states in multiple steps (see \S~\ref{sec:setup:singlenode} for an example exploit, $XT_6$) and are stealthy to detection. They can evade the mitigation designed for earlier and simpler ADAMS attacks. For instance, Geth $v1.11.4$, which is already patched against DETER attacks, is found still vulnerable by \textsc{mpfuzz}, such as under the $XT_6$ attack~\cite{wang2023understanding}.

Current defenses against mempool DoS attacks typically involve patching vulnerabilities after specific attack patterns have been identified. This reactive approach, while necessary, cannot fully guarantee eviction security for the mempool, as there may be undiscovered attacks capable of bypassing these patches. Therefore, the key to ensuring robust mempool security lies in a formal understanding of mempool DoS security. Without a precise security definition, it is impossible to validate or certify the soundness or completeness of a mempool's defenses against unknown attacks, let alone design new mempools with provable security. Notably, the bug oracles used in tools like \textsc{mpfuzz} optimize search efficiency but do not guarantee completeness, further highlighting the need for a comprehensive security framework.

Note that ADAMS attacks that exploit mempool admission policies are not the only means to cause under-utilized blocks in a victim blockchain. There are other DoS vectors, notably the computing-resource-exhaustion attacks targeting any blockchain subsystems prior to block validation. Known exhaustion strategies include running under-priced smart-contract instructions~\cite{DBLP:conf/ndss/0002L20,me:eip150,me:dosblockgas} or exploiting ``speculative’’ contract execution capabilities, such as the \texttt{eth\_call} in Ethereum’s RPC subsystem~\cite{DBLP:conf/ndss/LiCLT0L21} or censorship enforcement in Ethereum PBS (proposer-builder separation) subsystem (i.e., the ConditionalExhaust attack in~\cite{cryptoeprint:2023/956}). Resource exhaustion by adversarial smart contracts is out of the scope of this paper. Besides, denial of blockchain services have been studied across different layers in a blockchain system stack including eclipse attacks on the P2P networks~\cite{DBLP:conf/uss/HeilmanKZG15,DBLP:journals/iacr/MarcusHG18,DBLP:conf/sp/ApostolakiZV17,tran2020stealthier}, DoS blockchain consensus~\cite{mirkin2019bdos,me:finney}, DoS state storage~\cite{me:NURGLEDosstorage}, etc.

This work formulates the economic-security notions to protect mempools against asymmetric eviction DoS attacks. Economic security entails fee calculation, which may be non-deterministic in Ethereum and poses challenges in mempool designs. This work focuses on the mempool DoS by exploiting transaction admission policies, which normally rely only on deterministic transaction prices and are agnostic to non-deterministic Gas or fees. 

\noindent{\bf Proposed policy}: This work presents the first definitions of asymmetric eviction mempool DoS security and the provably secure mempool designs. Specifically, we conceive a general pattern of mempool eviction-based DoSes: Attackers aim to evict the victim transactions residential in a mempool. Based on the pattern, we formulate a security definition, that is, $g$-eviction security, which requires a secure mempool to lower bound the total fees of residential transactions under arbitrary sequences of transaction arrivals. 

This work presents \textsc{saferAd-CP}, a transaction admission policy for securing mempools against asymmetric eviction DoSes. \textsc{saferAd-CP} can achieve eviction security by lower-bounding the fees of any arriving transactions causing eviction which is much higher than original Geth. We focus on designing admission policy in a pending-transaction mempool, that is, the pool is supposed to store only valid transactions (whereas future or other invalid transactions are buffered in a separate pool).\footnote{In the rest of the paper, the term mempool refers only to the pool storing pending transactions.} 

\textsc{saferAd-CP} prevents valid residential transactions from being turned into invalid ones by evicting only ``child-less'' transactions, that is, the ones with the maximal nonce of its sender.
Upon each arriving transaction, if the transaction is admitted by evicting another transaction, it enforces that the fee of the admitted transaction must not be lower than a pre-set lower bound. 

\noindent{\bf Evaluation}:
We analyze \textsc{saferAd-CP} and prove the eviction-based DoS security with eviction bound: lower-bound the attack costs (i.e., the adversarial transaction fees) under eviction DoSes.

We implement a \textsc{saferAd-CP} prototype over Geth by evicting only ``child-less'' transactions. We collect real-world transactions from the Ethereum mainnet, and replay them against the \textsc{saferAd-CP} prototype. The evaluation shows \textsc{saferAd-CP} incurs a revenue change between $[+1.18\%, +7.96\%]$ under replayed real-world transaction histories. \textsc{saferAd-CP}'s latency in serving transactions is at most $7.3\%$ different from a vanilla Geth client. Under the attacks that latest Geth is known to be vulnerable for, \textsc{saferAd-CP} increases the attacker's cost by more than $10^4$ times.

\noindent{\bf Contributions}:
Overall, this paper makes the following contributions:

\vspace{2pt}\noindent$\bullet$\textit{ 
New security definitions}:
Presented the first economic-security definitions of mempools against asymmetric eviction DoSes. Specifically, we formulate a general pattern of mempool eviction DoSes and present a security definition to mitigate this attack pattern.

\vspace{2pt}\noindent$\bullet$\textit{ 
Proven secure designs}:
Presented the first eviction DoS-secure mempool designs, \textsc{saferAd-CP}, that achieves proven eviction-based security. The security stems from \textsc{saferAd-CP}'s design in lower-bounding the attack cost under eviction DoSes.

\vspace{2pt}\noindent$\bullet$\textit{ 
Performance \& utility evaluation}:
Implemented a \textsc{saferAd-CP} prototype on Geth and evaluated its performance and utility by replaying real-world transaction traces. It shows \textsc{saferAd-CP} incurs negligible overhead in latency and maintains a high lower bound on validator revenue under any adversarial traces.

\section{Background and Notations}
\noindent{\bf Transactions}: In Ethereum, a transaction $tx$ is characterized by a $sender$, a $nonce$, a $price$, an amount of computation it consumes, $GasUsed$, and the $data$ field relevant to smart-contract invocation. Among these attributes, transaction $sender$, $nonce$, and $price$ are ``static'' in the sense that they are independent of smart contract execution or the context of block validation (e.g., how transactions are ordered). This work aims at lightweight mempool designs leveraging only static transaction attributes. We denote an Ethereum transaction by its $sender$, $nonce$, and $price$. For instance, a transaction $tx_1$ sent from Account $A$, with nonce $3$, of price $7$ is denoted by $\langle{}A3, 7\rangle{}$.

A transaction $tx_1$ is $tx_2$'s ancestor or parent if $tx_1.sender=tx_2.sender\land{}tx_1.nonce<tx_2.nonce$. We denote the set of ancestor transactions to transaction $tx$ by $tx.ancestors()$. Given the transaction set in a mempool state $ts$, $tx$'s ancestor transactions in $ts$ and with consecutive nonces to $tx$ are denoted by set $tx.ancestors()\cap{}ts$. For instance, suppose transactions $tx_1,tx_2,tx_3, tx_4$ are all sent from Alice and are with nonces $1$, $2$, $3$, and $4$, respectively. Then, $tx_4.ancestors()\cap{}\{tx_1, tx_3\text{\}}=\{tx_3\text{\}}$.

A transaction $tx$ is a future transaction w.r.t. a transaction set $ts$, if there is at least one transaction $tx'\not\in{}ts$ and $tx'$ is an ancestor of $tx$. Given any future transaction $tx$ in set $ts$, we define function $\textsc{isFuture}(tx, ts)=1$.

\noindent{\bf Transaction fees}: A transaction $tx$'s fee is the product of $GasUsed$ and $price$, that is, $tx.fee=GasUsed\cdot{}price$. $GasUsed$ is determined by a fixed amount ($21000$ Gas) and the smart contract execution by $tx$. In Ethereum, The latter factor is sensitive to various runtime conditions, such as how transaction $tx$ is ordered in the blocks. After EIP-1559, part of price and fees are burnt; in this case, block revenue excludes the burnt part.

\noindent{\bf The notations} used in this paper (some of which are introduced later in the paper) are listed in Table~\ref{tab:notations}. Notably, we differentiate the unordered set and ordered list: Given an unordered transaction set, say $ops$, an ordered list of the same set of transactions is denoted by a vector, $\vec{ops}$. A list is converted to a set $ops$ by the function $unorder(\vec{ops})$, and a set is converted to a list $\vec{ops}$ by the function $toOrder(ops)$.

\begin{table}[!htbp] 
\caption{Notations: txs means transactions.}
\label{tab:notations}\centering{\footnotesize
\begin{tabularx}{0.49\textwidth}{ |l|X|l|X| }
  \hline
 & Meaning &  & Meaning  \\ \hline
$ops$ & Arriving txs
&
$\vec{ops}$ & List of arriving txs
\\ \hline
$st$ & Txs in mempool
&
$\vec{st}$ & List of txs in mempool
\\ \hline
$dc$ & Txs declined or evicted from mempool
&
$m$ & Mempool length 
\\ \hline
$\langle{}A3, 7\rangle{}$ 
& \multicolumn{3}{l|}{A tx of sender $A$, nonce $3$, and price $7$} 
\\ \hline
\end{tabularx}
}
\end{table}

\noindent{\bf Blockchain mempools}:
In blockchains, recently submitted transactions by users, a.k.a., unconfirmed transactions, are propagated, either privately or publicly, to reach one or multiple validator nodes. Mempool is a data structure a blockchain full node uses to buffer the ``unconfirmed'' transactions before these transactions are included in blocks. 

In Ethereum 2.0, transaction propagation follows two alternative paths. A public transaction is propagated to the entire network, while a private transaction is forwarded by a builder to the proposers who he has established the connection with. 
In both scenarios, the mempool faces the same design issues: Because the mempool needs to openly accept a potentially unlimited number of transactions sent from arbitrary EOA accounts, it needs to limit the capacity and enforce policies for transaction admission. In practice, we found the mempool storing public transactions has the same codebase as that storing private mempool.\footnote{For instance, the mempool in Flashbot builder~\cite{me:flashbotbuilder} used to store both private and public transactions is a fork (with no code change) of the mempool storing only public transactions in Geth.}

\section{Threat Model}
\label{sec:threat}

\begin{figure}
  \centering
  \includegraphics[width=0.425\textwidth]{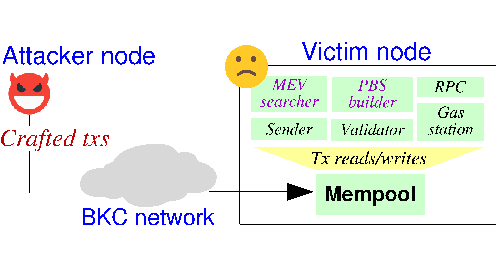}
  \caption{Threat model of a victim mempool: In blue are downstream operators that rely on reading or writing the mempool. In dark blue are the operators in the private transaction path. }
  \label{fig:downstream}
\end{figure} 

In the threat model, an adversary node is connected, either directly or indirectly, to a victim node on which a mempool serves various downstream operators reading or inserting transactions. Figure~\ref{fig:downstream} depicts some example mempool-dependent operators, including a transaction sender who wants to insert her transaction to the mempool, a full node that reads the mempool to decide whether a received transaction should be propagated, a validator or PBS builder that reads the mempool and selects transactions to be included in the next block, an MEV searcher that reads the mempool to find profitable opportunities, a Gas-station service that reads the mempool to estimate appropriate Gas price for sending transactions, etc. 

The adversary's goal is two-fold: 1) Deny the victim mempool's service to these critical downstream operators. This entails keeping all normal transactions out of the mempool so that the transaction read/write requests from operators would fail. 2) Keep the adversary's cost asymmetrically low. This entails keeping the transactions sent by the adversary away from being included in the blockchain. 

The adversary has the capacity to craft transactions and send them to reach the victim mempool. In the most basic model, the adversary is directly connected to the victim node. In practice, the adversary may launch a super-node connected to all nodes and aim at attacking all of them or selecting critical nodes to attack (as done in DETER~\cite{DBLP:conf/ccs/LiWT21}). Alternatively, the adversary may choose to launch a ``normal'' node connected to a few neighbors and propagate the crafted transactions via the node to reach all other nodes in the network.

\section{Security Definition}
\label{sec:secdef}

To start with, we characterize two mempool states at any time: the set of transactions residing in the mempool denoted by $st$, and the set of transactions declined or evicted from the mempool denoted by $dc$. We also characterize the list of confirmed transactions included in produced blocks by $\vec{bks}=\{\vec{b}\}$ and the list of transactions arriving by $\vec{ops}$.

A mempool commonly supports two procedures: transaction admission and block building. We specify the first one, which is relevant to this work. 

\begin{definition}[Tx admission]
\label{def:txadmit}
Given an arriving transaction $tx_i$ at a mempool of state $st_i$, an admission algorithm $\textsc{adTx}$ would transition the mempool into an end state $st_{i+1}$ by admitting or declining $tx_i$ or evicting transactions $te_i$. Formally,

\begin{equation}
\textsc{adTx}(st_i, tx_i) \rightarrow{} st_{i+1}, te_i
\end{equation}
\end{definition}

\begin{definition}[Tx admission timeline]
\label{def:timeline}
In a transaction admission timeline, a mempool under test is initialized at state $\langle{}st_0, dc_0=\emptyset\rangle{}$, receives a list of arriving transactions in $\vec{ops}$, and ends up with an end state $\langle{}st_n, dc_n\rangle{}$. Then, a validator continually builds blocks from transactions in the mempool $st_n$ until it is empty, leading to eventual state $\langle{}st_l=\emptyset, dc_l=dc_n\rangle{}$ and newly produced blocks $\vec{bks_l}$ with transactions in $st_n$. This transaction-admission timeline is denoted by $f(\langle{}st_0, \emptyset\rangle{},\vec{ops})\Rightarrow{}\langle{}st_n, dc_n\rangle{}$. 
\end{definition}

As in the above definition, this work considers the timeline in which block arrival or production does not interleave with the arrival of adversarial transactions. We leave it to the future work for interleaved block arrival and attacks. In the following, we define the mempool security against asymmetric attacks.

\begin{definition}[Mempool eviction security]
\label{def:security:evict}
Consider any mempool that initially stores normal transactions $st_0$, receives a list of benign and adversarial transactions $\vec{ops}$ and converts to the end state $st_n$. 

The mempool is secure against asymmetric eviction attacks, or {\bf $g$-eviction-secure}, if.f. the total transaction fees under any adversarial transactions are higher than $g(st_0)$ where $g(\cdot)$ is a price function that takes as input a mempool state and returns a price value. Formally,

\begin{eqnarray}
\nonumber 
&& \forall{} st_n, \vec{ops},  
\exists \text{function } g(\cdot)
\\
&& \text{ s.t., } \forall{st_0}, fees(st_n) \geq{} g(st_0)
\end{eqnarray}
\end{definition}

The $g$-eviction-security definition ensures that under arbitrary attacks, the total fees of transactions inside the mempool are lower bound by a value dependent only on normal transactions in initial state. 
This definition reflects the following intuition: If the normal transactions in initial state can provide enough block revenue for validators, the $g$-security of mempool ensures the mempool under any adversarial workloads (i.e., attacks) have enough fees or enough block revenue for validators.

\begin{definition}[Mempool locking security]
\label{def:security:locking}
Consider under attacks in which the mempool of initial state storing only benign transactions $\langle{}st_0, \emptyset\rangle{}$ receives a transaction sequence $\vec{ops}$ interleaved of adversarial and benign transactions, transitions its state, and reaches the end state $\langle{}st_{n}, dc_n\rangle{}$, as shown Equation~\ref{eqn:def:attacktimeline2}. 

The mempool is said to be secure against asymmetric locking attacks, or {\bf $h$-locking-secure}, if.f. the maximal price of transactions declined or evicted in the end state under attacks, i.e., $dc_n$, is lower than a certain price function $h()$ on the end-state mempool under attacks, i.e., $h(st_n)$. It is required that for any mempool state $st$, $h(st)$ must be lower than the average transaction price in $st$, namely $\forall{}st, h(st)<avgprice(st)$. Formally,

\begin{eqnarray}
\nonumber 
&& \forall{} st_0, \vec{ops},  
\label{eqn:def:attacktimeline2} f(\langle{}st_0, \emptyset\rangle{}, \vec{ops})\Rightarrow{} \langle{}st_n, dc_n\rangle{}
\\
\nonumber 
&&\exists \text{function } h(\cdot) < avgprice(\cdot)
\\
&& \text{ s.t., } maxprice(dc_n) < h(st_{n})
\end{eqnarray}
\end{definition}

The $h$-locking-security definition ensures that the maximal price of declined or evicted transactions from the mempool is upper bound by that of the transaction staying in the mempool.

\section{\textsc{saferAd} Framework}

\subsection{Tx-Admission Algorithm of $CP$}
\begin{algorithm}[h]
\caption{\textsc{saferAd-CP}(MempoolState $st$, Tx $ta$)}
\small \label{alg:adTx:AA:framework}
\begin{algorithmic}[1]
\If{$\textsc{preCkV}(ta, st)==0$} \Comment{Precheck tx validity} 
\label{alg:adTx:AA:0}
  \State \Return $te=ta$; \Comment{Decline invalid $ta$} 
\EndIf
\State $\vec{tes}$=\textsc{torder}($st.findChildless(), price$); 
\label{alg:adTx:AA:1}
\State $te$=$\vec{tes}$.lastTx();
\label{alg:adTx:AA:2}
\If{$\|st\| == m$}
\If{$ta.price \leq te.price $} \label{alg:adTx:AA:3}
 \State \Return $te=ta$; \Comment{Decline $ta$} 
\label{alg:adTx:AA:4}
\Else
  \State $st$.admit($ta$).evict($te$); 
\label{alg:adTx:AA:5}
\EndIf
\Else
  \State $st$.admit($ta$); \label{alg:adTx:AA:8}
  \State $te=$NULL; \label{alg:adTx:AA:9}
\EndIf
\end{algorithmic}
\end{algorithm}

The proposed Algorithm~\ref{alg:adTx:AA:framework} enforces the invariant that each admitted transaction to a mempool evicts at most one transaction from the mempool. 

The algorithm initially maintains a mempool of state $st$ and receives an arriving transaction $ta$. The algorithm decides whether and how to admit $ta$ and produces as the output the end state of the mempool and the evicted transaction $te$ from the mempool. In case that $ta$ is declined by the mempool, $te=ta$.

Internally, the algorithm first pre-checks the validity of $ta$ on mempool $st$ (in Line~\ref{alg:adTx:AA:0}). If $ta$ is a future transaction or overdrafts its sender balance, the transaction is deemed invalid and is declined from entering the mempool. The algorithm only proceeds when $ta$ passes the validity precheck.

It then sorts all childless transactions in mempool ($st.findChildless()$) in descendant order based on tx's price. As described next, function $\vec{tes}=\textsc{torder}(st)$ produces a order of transactions in the mempool, denoted by $\vec{tes}$.
It selects the last transaction on this ordered list, denoted by $te$. That is, $te$ has the lowest score on $\vec{tes}$. 

If the mempool is full and $ta.price \leq te.price$ (Line~\ref{alg:adTx:AA:3}), the algorithm declines the arriving transaction $ta$, that is, $te=ta$ (Line~\ref{alg:adTx:AA:4}). Otherwise, if the mempool is full and $ta.price > te.price$, the algorithm admits $ta$ and evicts $te$ (Line~\ref{alg:adTx:AA:5}).

If the mempool is not full, the algorithm always admits $ta$ to take the empty slot (Line~\ref{alg:adTx:AA:8}).

\subsection{Eviction Secure of $CP$}
Policy \textsc{CP} achieves eviction security in the sense that it ensures the total prices monotonically increase. 

Suppose a mempool runs Algorithm~\ref{alg:adTx:AA:framework} and transitions from state $st_i$ to $st_{i+1}$. No transaction in $st_i$ can be turned into a future transaction in $st_{i+1}$.
Because the eviction candidates can only be the childless transactions (Recall Line~\ref{alg:adTx:AA:1} in Algorithm~\ref{alg:adTx:AA:framework}).



\begin{lemma}[Monotonic price-increasing]
\label{thm:pr:mono} 
If a mempool runs Algorithm~\ref{alg:adTx:AA:framework}, the sum of transaction prices in the mempool monotonically increases, or the mempool is considered to be monotonic price-increasing. Formally,

\begin{equation}
\label{eqn:iterative:n}
\sum_{tx\in{}st_n} tx.price \geq \sum_{tx\in{}st_0} tx.price
\end{equation}
\end{lemma}

\begin{proof}
Consider that a mempool running Algorithm~\ref{alg:adTx:AA:framework} receives an arriving transaction $ta_i$ and transitions from state $st_i$ to $st_{i+1}$. We aim to prove that,
\begin{equation}
\label{eqn:pr:iterative:i}
\sum_{tx\in{}st_{i+1}} tx.price \geq \sum_{tx\in{}st_i} tx.price
\end{equation}

Now we consider three cases for state transition: 
1) $ta_i$ is declined, 
2) $ta_i$ is admitted by taking an empty slot in $st_i$ (i.e., no transaction is evicted), and 
3) $ta_i$ is admitted by evicting $te_i$.

In Case 1), $st_i=st_{i+1}$. Thus,
\begin{eqnarray}
\nonumber
\sum_{tx\in{}st_i} tx.price
&=& \sum_{tx\in{}st_{i+1}} tx.price
\end{eqnarray}

In Case 2), $st_{i+1}=\{ta_i$\text{\}}$\cup{}st_i$. Thus,

\begin{eqnarray}
\nonumber
\sum_{tx\in{}st_{i+1}} tx.price &=&  
\sum_{tx\in{}st_i\cup{}\{ta_i\text{\}}} tx.price 
\\ \nonumber &=& 
\sum_{tx\in{}st_i} tx.price  + ta_i.price 
\\ \nonumber &\geq{}& 
\sum_{tx\in{}st_i} tx.price
\end{eqnarray}

In Case 3), we denote by $st$ the set of transactions in both $st_i$ and $st_{i+1}$. That is, $st=st_i\cap{}st_{i+1}$. In Case 3), we have $st_{i+1}\setminus{}\{ta_i$\text{\}}$=st_i\setminus{}\{te_i$\text{\}}$=st$. Applying Line~\ref{alg:adTx:AA:3} in Algorithm~\ref{alg:adTx:AA:framework}, we can derive the following:

\begin{eqnarray}
\nonumber
\sum_{tx\in{}st_{i+1}} tx.price &=& 
\sum_{tx\in{}st\cup{}\{ta_i\text{\}}} tx.price
\\ \nonumber &=& 
\sum_{tx\in{}st} tx.price + ta_i.price
\\ \nonumber &\geq{}& 
\sum_{tx\in{}st} tx.price + te_i.price 
\\ \nonumber &=& 
\sum_{tx\in{}st\cup{}\{te_i\text{\}}} tx.price 
\\ \nonumber &=& 
\sum_{tx\in{}st_i} tx.price
\end{eqnarray}

Therefore, in all three cases, Equation~\ref{eqn:pr:iterative:i} holds. In general, for any initial state $st_0$ and any end state $st_n$ that is transitioned from $st_0$ with $i\in[0,n-1]$, one can iteratively apply Equation~\ref{eqn:pr:iterative:i} for $i\in[0,n-1]$ and prove the sum of transaction price monotonically increases. 
\end{proof}

\begin{theorem}[$g$-eviction security of Algorithm~\ref{alg:adTx:AA:framework}]
\label{thm:policya:eviction}
A mempool running Algorithms~\ref{alg:adTx:AA:framework} is $g$-eviction secure. That is, given a sequence of arriving normal and adversarial transactions $\vec{ops}$ to the mempool state, the total transaction fees in the end-state mempool under attacks $st_n$ are lower-bounded by the following:

\begin{eqnarray}
\forall st_0, \vec{ops}, fees(st_n) &\geq{}& g(st_0) 
\label{eqn:mf:bound}
\end{eqnarray}
\end{theorem}

\begin{proof}
Due to Lemma~\ref{thm:pr:mono}, we prove the property of increasing sum of transaction prices during the mempool state transition, that is, Equation~\ref{eqn:iterative:n}. 

And the minimal gas for each transaction is $21000$. We have a lower bound for transaction $fee$ which is $21000\cdot{}tx.price$, that is,

\begin{eqnarray}
\nonumber
& fees(\vec{st_n}) \geq{} 21000\cdot{}\sum_{tx\in{}st_n}tx.price & 
\\ 
\nonumber
& \geq{} 21000\cdot{}\sum_{tx\in{}st_n} tx.price & 
\\ 
\label{eqn:cp:bound}
& \geq{} 21000\cdot{}\sum_{tx\in{}st_0} tx.price & 
\end{eqnarray}

Equation~\ref{eqn:mf:bound} holds.
\end{proof}

\subsection{Locking Insecure of $CP$}



Policy \textsc{saferAd-CP} is not locking secure. We illustrate this with a counterexample that highlights how an attacker can exploit policy $CP$ to achieve a denial-of-service (DoS) with minimal cost.

Consider a mempool implementation that imposes no limitations on the number of transactions sent by a single sender. Suppose the mempool is at full capacity, and all transactions originate from a single sender. These transactions are represented as: $tx_1, tx_2, \ldots, tx_{n-1}$, each with a price of 1, while the only transaction without child dependencies, $tx_n$, has the highest nonce and a significantly higher price of 10,000.

The mempool operates under policy $CP$ where an incoming transaction $tx_a$ can only evict an existing transaction if $tx_a$ has a higher price than the transaction it aims to evict. Given this policy, $tx_a$ would need a price greater than 10,000 to evict $tx_n$. This creates a scenario where the attacker can effectively lock the mempool with a low-cost strategy by ensuring $tx_n$ remains in place, thus preventing the eviction of other transactions priced at 1.


\section{Implementation Notes on Geth}

We build a prototype implementation of \textsc{saferAd-CP} on Geth $v1.11.4$. We describe how the pending mempool in vanilla Geth handles transaction admission and then how we integrate \textsc{saferAd-CP} into Geth.

\noindent{\bf
Background: Geth mempool implementation}: In Geth $v1.11.4$, the mempool adopts the price-only Policy patched with extra checks. Concretely, upon an arriving transaction $ta$, \circled{0} Geth first checks the validity of $ta$ (e.g., $ta$ is an overdraft), then it checks if the mempool is full. \circled{1} If so, it finds the transaction with the lowest price as the candidate of eviction victim $te'$. \circled{2} It then removes $te'$ from the primary storage and the secondary index. \circled{3} At last, it adds $ta$ to the primary storage and the secondary index.
If the mempool is not full, \circled{4} Geth adds $ta$ to the primary storage and to the secondary index.

For fast transaction lookup, Geth $v1.11.4$ maintains two indices to store mempool transactions (i.e., each transaction is stored twice): a primary index where transactions are ordered by price and a secondary index where transactions are ordered first by senders and then by nonces. 

In Step \circled{0}, we adopt the patch against MemPurge attacks~\cite{me:mempurge:fix}.

\noindent{\bf
Implementation of \textsc{saferAd-CP} on Geth}: Geth's mempool architecture is well aligned with Algorithm~\ref{alg:adTx:AA:framework}.
We overwrite Step \circled{1} in Geth; instead of finding the transaction with the lowest price, we find the childless transaction with the lowest price, which is to scan Geth's price-based index from the bottom, that is, the transaction with the lowest price; for each transaction, we check if the transaction has a defendant in the mempool by querying the second index. If so, we continue to the transaction above in the primary index and repeat the check. If not, we select the transaction to be eviction victim $te$. Steps \circled{2}, \circled{3}, and \circled{4} remain the same except that reference $te'$ is replaced with $te$.

\section{Revenue Evaluation}

\subsection{Experimental Setups}
\label{sec:setup}

\noindent{\bf Workload collection}: 
For transaction collection, we first instrumented a Geth client (denoted by Geth-m) to log every message it receives from every neighbor. The logged messages contain transactions, transaction hashes (announcements), and blocks. When the client receives the same message from multiple neighbors, it logs it as multiple message-neighbor pairs. We also log the arrival time of a transaction or a block.

We ran a Geth-m node in the mainnet and collected transactions propagated to it from Sep. 5, 2023 to Oct. 5, 2023. In total, $1.5*20^8$ raw transactions are collected, consuming $30$ GB-storage. 
We make the collected transactions replayable as follows: We initialize the local state, that is, account balances and nonces by crawling relevant data from infura.io. We then replace the original sender in the collected transactions with the public keys that we generated. By this means, we know the secret keys of transaction senders and are able to send the otherwise same transactions for experiments. 

We choose $8$ traces of consecutive transactions from the raw dataset collected, each lasting $2.5$ hour. We run experiments on each $2.5$-hour trace. The reason to do so, instead of running experiments directly on the one-month transaction trace, is that the initialization of blockchain state in each trace requires issuing RPC queries (e.g., against infura) on relevant accounts, which is consuming; for a $2.5$-hour trace, the average time of RPC querying is about one day. To make the selected $2.5$-hour traces representative, we cover both weekdays and weekends, and on a single day, daytime and evening times.

\begin{figure}[!bthp]
\centering
\includegraphics[width=0.375\textwidth]{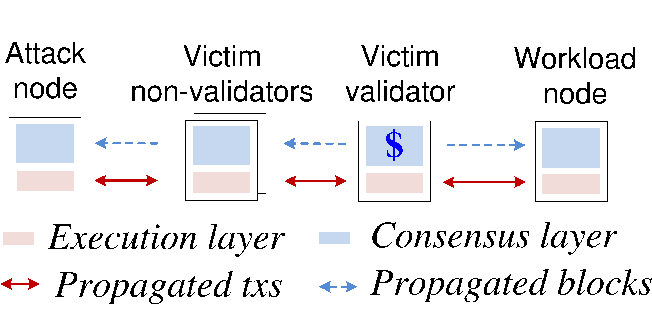}
\caption{Experimental setup}
\label{fig:exp:setup}
\end{figure}

\noindent{\bf Experimental setup}:
For experiments, we set up four nodes, an optional attack node sending crafted transactions, a workload node sending normal transactions collected, a victim non-validator node propagating the transactions and blocks between attack node and victim validator node, and a victim validator node receiving transactions from the workload node and attack node through non-validator node. The victim validator node is connected to both the workload node and victim non-validator node which is connected to the attack node. There is no direct connection between the attack node and the workload node. The attack node runs an instrumented Geth $v1.11.4$ client (denoted by Geth-a) that can propagate invalid transactions to its neighbors. The victim nodes runs the target Ethereum client to be tested; we tested two victim clients: vanilla Geth $v1.11.4$ and \textsc{saferAd-CP} $CP$; the latter one is implemented as addon to Geth $v1.11.4$. The workload node runs a vanilla Geth $v1.11.4$ client. On each node, we also run a Prysm $v3.3.0$ client at the consensus layer. The experiment platform is depicted in Figure~\ref{fig:exp:setup}.
Among the four nodes, we stake Ether to the consensus-layer client on the victim validator node, so that only the victim validator node would propose or produce blocks. 

We run experiments in two settings: under attacks and without attacks. For the former, we aim at evaluating the security of \textsc{saferAd-CP} under attacks, that is, how successful DoS attacks are on \textsc{saferAd-CP}. In this setting, we run all three nodes (the victim, attacker and workload nodes). For the latter, we aim at evaluating the utility of \textsc{saferAd-CP} under normal transaction workloads. In this setting, we only run victim and workload nodes, without running the attack node.

In each experiment, 1) we replay the collected transactions as follows: For each original transaction $tx$ collected, we send a replayed transaction $tx'$ by replacing its sender with a self-generated blockchain address. $GasUsed$ is simulated: If $tx$ runs a smart contract, $tx'$ does not run the same contract. Instead, we make $tx'$ run our smart contract with $tx'.GasUsed$ equals $tx.GasUsed$ if $tx$ is included in the mainnet blockchain, or equal $tx.Gas$ (i.e., the Gas limit set by the transaction sender) if $tx$ is not included.
2) When replaying a collected block, we turn on the block-validation function in Prysm, let it produce and validate one block, and send the block (which should be different from the content of the collected block) to the Geth client. We then immediately turn off block validation before replaying the next transaction in the trace.

\subsection{Revenue Under Normal Transactions}

This experiment compares different mempool policies by evaluating their block revenue under the same transaction workloads. 

\noindent{\bf Experimental method}: We consider two mempool policies, the baseline one in Geth $v1.11.4$ and \textsc{saferAd}-\textsc{CP} on the baseline. Given each policy, we replay the $8$ transaction traces in the same way as before and collect the produced blocks. We report the average revenue per block collected from the blocks. 

\noindent{\bf Results}: Figure~\ref{fig:Revenue_block_real_trace} presents the revenue of the selected $150$ consecutive blocks from the $600$ blocks in Trace $2$. The numbers of the three mempool policies are close, and they fluctuate in a similar way. 
Table~\ref{tab:revenue:8traces} presents the aggregated results by the total revenue and revenue per block. Compared to Geth $v1.11.4$, Policy \textsc{CP}'s revenue per block falls in the range of $[100\%+1.18\%, 100\%+7.96\%]$. This result suggests \textsc{saferAd-CP} incurs no significant change of block revenue under normal transactions.

\begin{figure}[htb]
  \centering
  \includegraphics[width=0.38\textwidth]{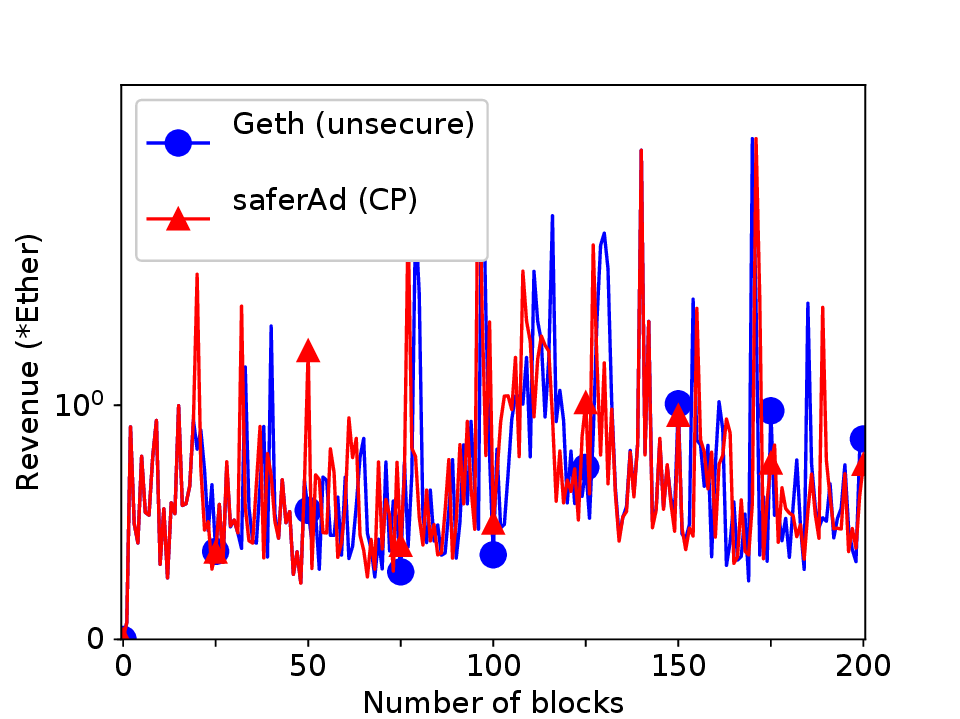}
  \caption{Block revenue w/w.o. \textsc{saferAd-CP} under Trace $1$.}
\label{fig:Revenue_block_real_trace}
\end{figure}

\begin{table}[!htbp]
\caption{Average block revenue (Ether) of different mempool policies. In bold are max and min numbers.}
\label{tab:revenue:8traces}
\centering{\scriptsize
\begin{tabularx}{0.4\textwidth}{|c|X|c| }
\hline 
Trace &  \textsc{CP} (Ether) 
& Geth (Ether)\\ \hline
1  
&$0.73\ (+3.13\%)$
&$0.71$	 \\ \hline
2
&$1.17\ (+6.97\%)$
&$1.09$	\\ \hline
3 
&$0.85\ (+4.55\%)$
&$0.81$	 \\ \hline
4
&$0.54\ (+4.56\%)$
&$0.52$	 \\ \hline   
5 
&$0.78\ (+6.66\%)$
&$0.73$	 \\ \hline
6 
&$1.19\ (+4.23\%)$
&$1.14$	 \\ \hline
7   
&\textbf{0.56 (+1.18\%)}
&$0.55$	 \\ \hline
8 
&\textbf{0.64 (+7.96\%)}
&$0.59$	 \\ \hline
\end{tabularx}
}
\end{table}


\subsection{Revenue Under Attacks}
\label{sec:setup:singlenode}

\noindent{\bf Background of attacks}: This experiment evaluates the security of \textsc{saferAd-CP} against known attacks. Given that our implementation is on Geth $v1.11.4$, we choose the attacks still effective on this version: $XT_6$~\cite{wang2023understanding}. Briefly, the attack works in four steps: 1) It first evicts the Geth mempool by sending $384$ transaction sequences, each of $16$ transactions from a distinct sender. The transaction fees are high enough to evict normal transactions initially in the mempool. 2) It then sends $69$ transactions to evict $69$ parent transactions sent in step 1) and turn their child transactions into future transactions. 3) Since now there are more than $5120$ pending transactions in the mempool, Geth's limit of $16$ transactions per send is off. It then conducts another eviction; this time, it sends all $5120$ transactions from one sender, evicting the ones sent in the previous round. 4) At last, it sends a single transaction to turn all transactions in the mempool but one into future transactions. The overall attack cost is low, costing the fee of one transaction.

Because $XT_6$ can evict Geth's mempool to be left with one transaction, we use as Geth's bound the maximal price in a given mempool times the maximal Gas per transaction (i.e., block Gas limit, namely $30$ million Gas).

\noindent{\bf Experimental method}: We set up the experiment platform described in \S~\ref{sec:setup:singlenode}. In each experiment, we drive benign transactions from the workload node. Note that the collected workload contains the timings of both benign transactions and produced blocks. On the $30$-th block, we start the attack. The attack node observes the arrival of a produced block and waits for $d$ seconds before sending a round of crafted transactions. 

The attack phase lasts for $36$ blocks; after the $66$-th block, we stop the attack node from sending crafted transactions. We keep running workload and victim nodes for another $24$ blocks and stop the entire process at the $90$-th block. We collect the blocks produced and, given a block, we report the total fee of transactions included. 

\begin{figure}[!ht]
  \centering
  \subfloat[$XT_6$ w. $8$-sec. delay (single node)]{%
\includegraphics[width=0.225\textwidth]{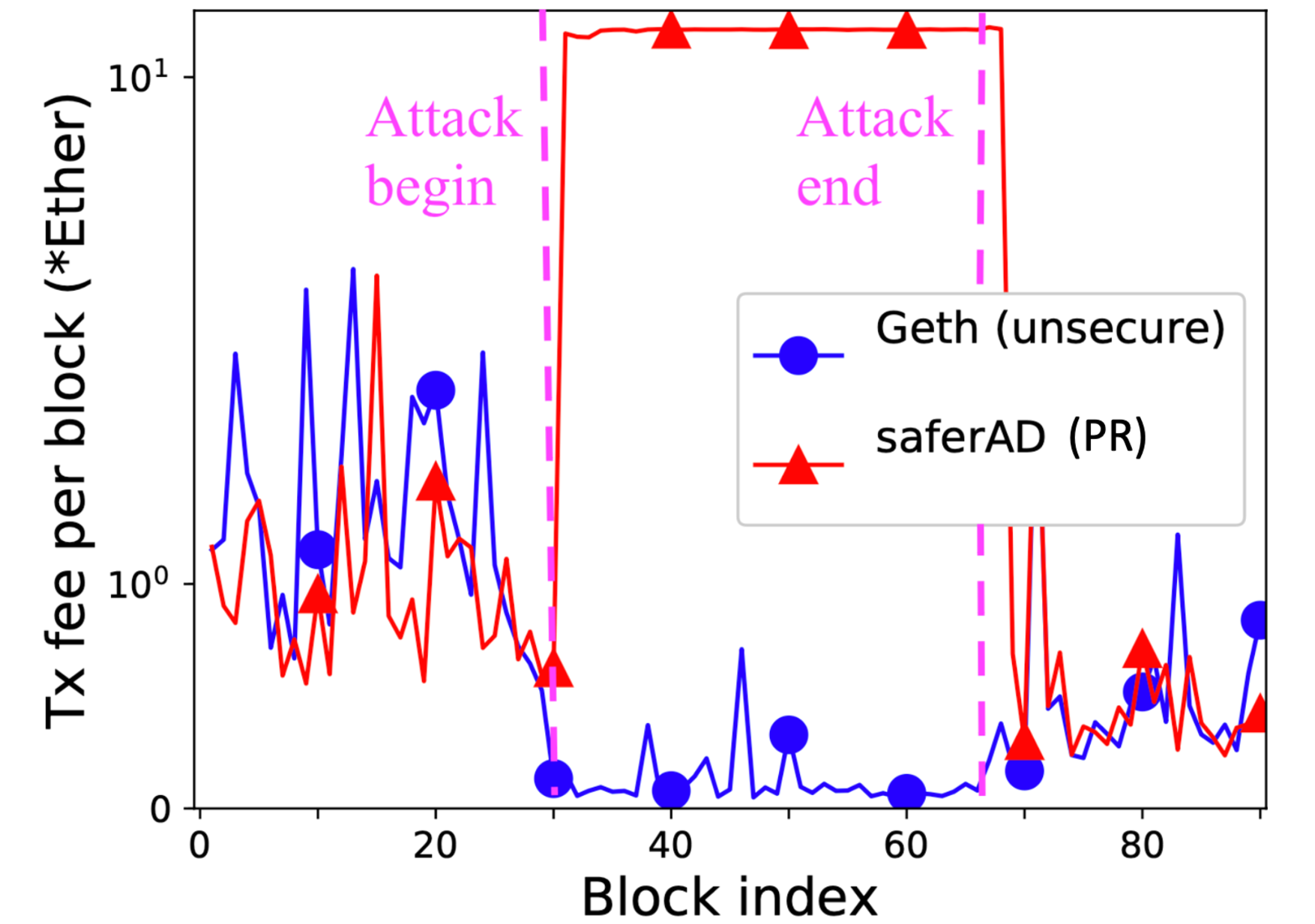}
    \label{fig:saferad-attack-eval}
    \label{fig:successrate-geth}}%
  \subfloat[$XT_6$ w. varying delay (single node)] {%
   \includegraphics[width=0.225\textwidth]{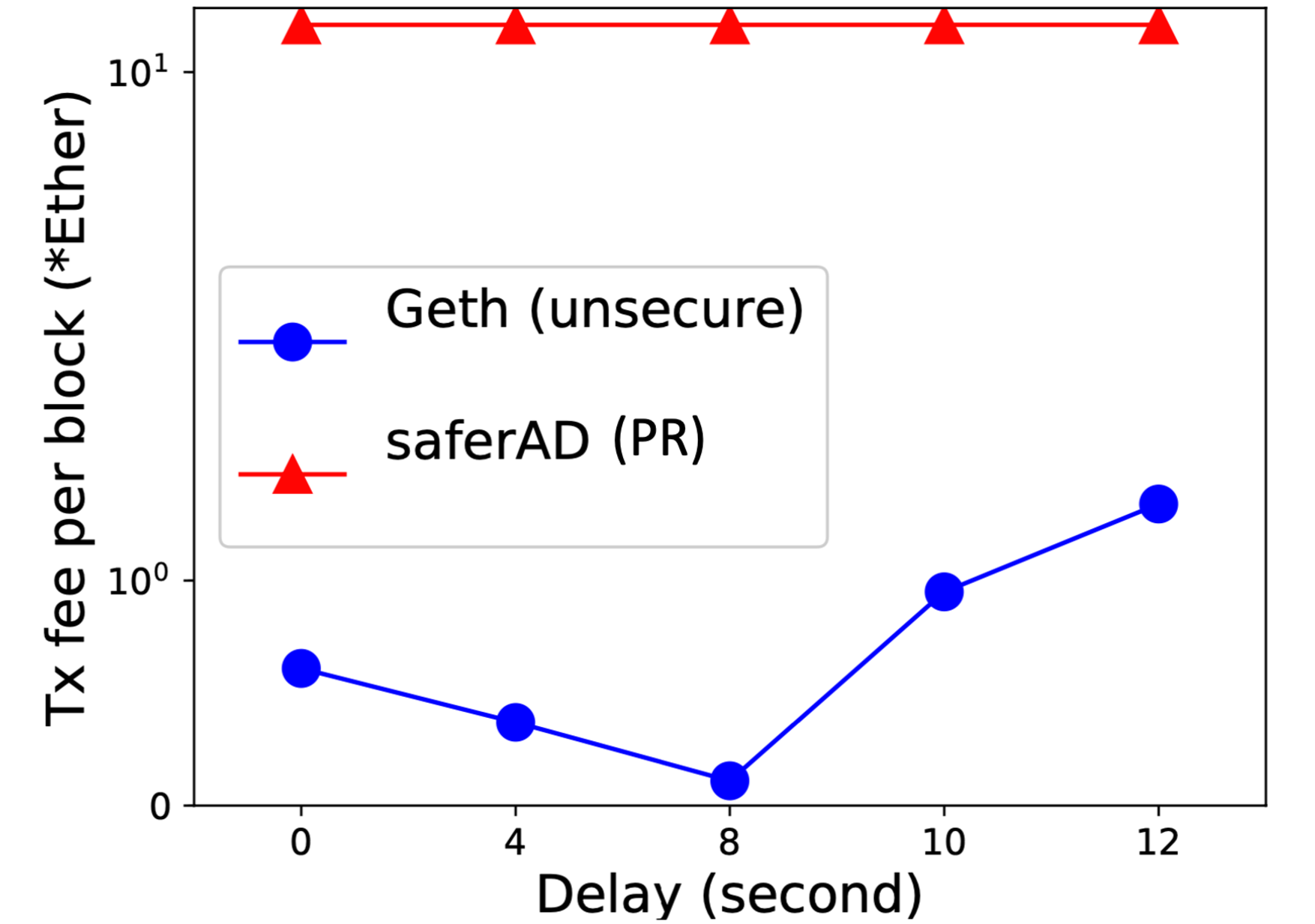}
  \label{fig:delay-geth}}
  \caption{Validator revenue under attacks: With and without \textsc{saferAd-CP} defenses.}%
\end{figure}

\noindent{\bf Results}: Figure~\ref{fig:saferad-attack-eval} reports the metrics for $XT_6$ on two victim clients: vanilla Geth $v1.11.4$ and \textsc{saferAd-CP}. Before the attack is launched (i.e., before the $30$-th block), the two clients produce a similar amount of fees for benign transactions, that is, around $0.7 - 2.0$ Ether per block. As soon as the attack starts from the $30$-th block, Geth's transaction fees quickly drop to zero Ether, which shows the success of $XT_6$ on unpatched Geth $v1.11.4$. There are some sporadic spikes (under $0.7$ Ether per block), which is due to that $XT_6$ cannot lock the mempool on Geth. 
Patching Geth with \textsc{saferAd-CP} can fix the vulnerability. After the attack starts on the $30$-th block, the transaction fees, instead of decreasing, actually increase to a large value; the high fees are from adversarial transactions and are charged to the attacker's accounts. The high fees show the effectiveness of \textsc{saferAd-CP} against known eviction-based attacks ($XT_6$).

Figure~\ref{fig:delay-geth} shows the total fee per block under $XT_6$ with varying delays, where the delay measures the time between when a block is produced and when the next attack arrives.
The results of Geth $v1.11.4$ show that with a short delay, the total transaction fees in mempool are high because $XT_6$ cannot lock a mempool, and the short block-to-attack delay leaves enough time to refill the mempool. With a median delay (e.g., sending an attack $8$ seconds after a block is produced), the attack is most successful, leaving mempool at zero Ether. With a long delay, the in-mempool transaction fees grow high due to the attack itself being interrupted by the block production. Under varying delays, the transaction fees of the mempool practicing \textsc{saferAd-CP} policy would remain constant at a high value, showing the defense effectiveness against attacks of varying delays.

\subsection{Estimation of Eviction Bounds}
\label{sec:eval:eviction:bounds}

This experiment estimates the eviction bounds of different admission policies under real-world transaction workloads. 

\noindent{\bf Experimental method}: In the experiment, we replayed the transaction traces we collected in \S~\ref{sec:setup} each of two Ethereum clients, be it either $CP$, or vanilla Geth $v1.11.4$ . In each run, right after producing each block, say $bk_i$, we record the mempool snapshot $st_i$. Then, assuming an attack starts right after the block $bk_i$ is produced and lasts for the next $10$ blocks, we estimate the lower bound of fees in the mempool right after block $bk_{i+10}$ under arbitrary eviction attacks.

1) For $CP$, we use Equation~\ref{eqn:cp:bound} as the estimation bound.
2) For the baseline, we consider vanilla Geth $v1.11.4$ under $XT_6$; instead of mounting actual attacks (which is time-consuming), we estimate the attack damage by considering that the mempool under attack contains only one transaction, which is consistent with the mempool end-state under an actual $XT_6$ attack.

\begin{figure}[htb]
  \centering
  \includegraphics[width=0.38\textwidth]{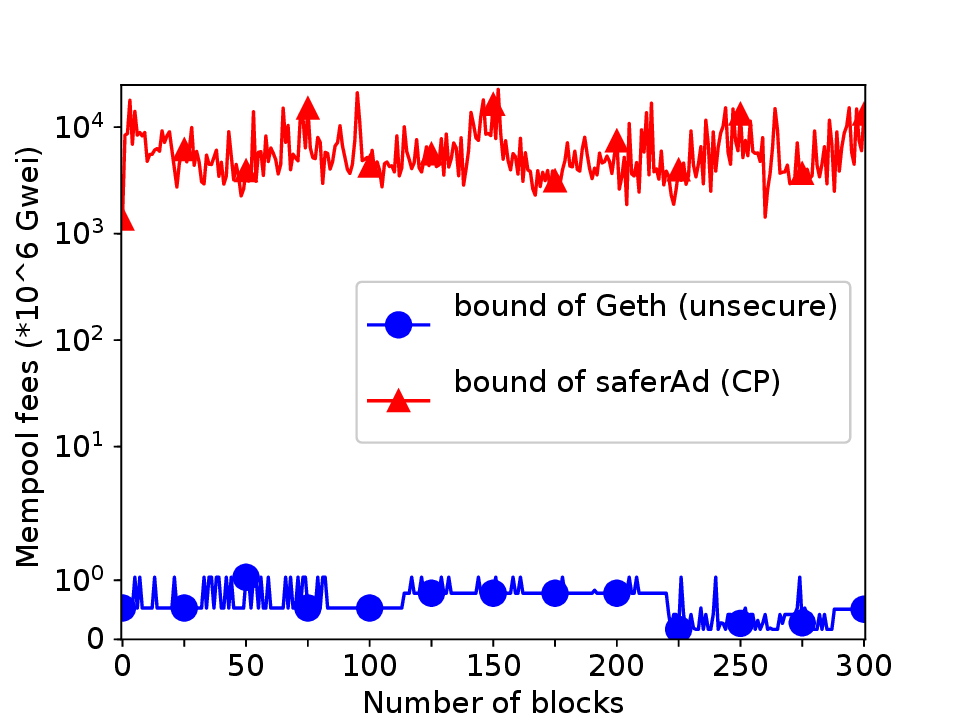}
  \caption{Block revenue lower bounds.}
\label{fig:Block_revenue_lower_bound_short}
\end{figure}

\begin{table}[!htbp]
\caption{Average, $95$-th, and $5$-th percentile of eviction bounds under different policies on eight transaction traces}
\label{tab:eviction:bound}
\centering{\small
\begin{tabularx}{0.4\textwidth}{|l|X|X|X|}
\hline 
  & Avg. bound (Ether) & $95\%$ bound (Ether) & $5\%$ bound (Ether) \\ \hline
$CP$ & $6.05$	&$13.72$	&$2.63$ \\ \hline
Geth & $0.63\cdot{}10^{-3}$	&$1.05\cdot{}10^{-3}$	&$0.17\cdot{}10^{-3}$ \\ \hline
\end{tabularx}
}
\end{table}

\noindent{\bf 
Results}: 
Figure~\ref{fig:Block_revenue_lower_bound_short} presents the results of estimated bounds over time. Compared to the mempool fees post attacks in Geth $v1.11.4$, both \textsc{saferAd} policies achieve high eviction bounds. The bound of $CP$ is higher than that of Geth $v1.11.4$ by $4$ orders of magnitude.
Table~\ref{tab:eviction:bound} presents statistics of the estimated bounds on the two clients tested. It includes the average, $95$-th, and $5$-th percentile of eviction bounds. \textsc{saferAd-CP} achieves statistically higher bounds than the baseline of Geth. On average, $CP$'s eviction bound is $6.05$ Ether, which is much higher than the $0.63\cdot{}10^{-3}$ Ether in Geth under attacks.
And specifically, for $CP$, 5\% of its bounds exceed $13.72$ Ether, and 95\% exceed $2.63$ Ether.

\section{Performance Evaluation}
\subsection{Platform Setup}

In this work, we use Ethereum Foundation's test framework~\cite{me:txpooltest} to evaluate the performance of the implemented countermeasures on Ethereum clients. Briefly, the framework runs two phases: In the initialization phase, it sets up a tested Ethereum client and populates its mempool with certain transactions. In the workload phase, it runs multiple ``rounds'', each of which drives a number of transactions generated under a target workload to the client and collects a number of performance metrics (e.g., latency, memory utilization, etc.). In the end, it reports the average performance by the number of rounds.
In terms of workloads, we use one provided workload (i.e., ``Batch insert'') and implement our own workload (``Mitigated attacks'').

\subsection{Experimental Results}

\noindent{\bf Workload ``Batch insert''}: 
We use the provided workload "Batch insert" that issues \texttt{addRemote} calls to the tested \texttt{txpool} of the Geth client. We select this workload because all our countermeasures are implemented inside the \texttt{addRemote} function. In the initialization phase, this workload sends no transaction to the empty mempool. In the workload phase, it runs one round and sends $n_0$ transactions from one sender account, with nonces ranging from $1$ to $n_0$, and of fixed Gas price $10000$ wei. When $n_0$ is larger than the mempool size, Geth admits the extra transactions to the mempool, buffers them, and eventually deletes them by running an asynchronous process that reorg the mempool (i.e., Function \texttt{pool.scheduleReorgLoop()} in Geth).

\begin{figure}[!ht]
\centering
\subfloat[Workload ``Batch insert'']{%
    \includegraphics[width=0.25\textwidth]{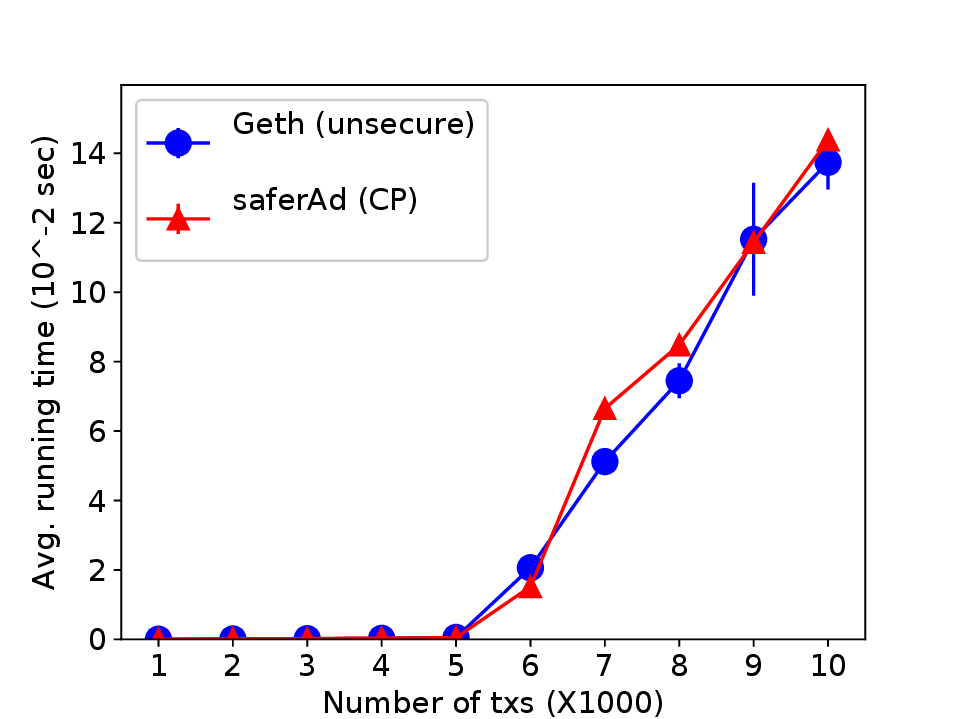}
    \label{fig:mentor-case}}
\subfloat[Workload ``Mitigated attack TN1'']{%
    \includegraphics[width=0.25\textwidth]{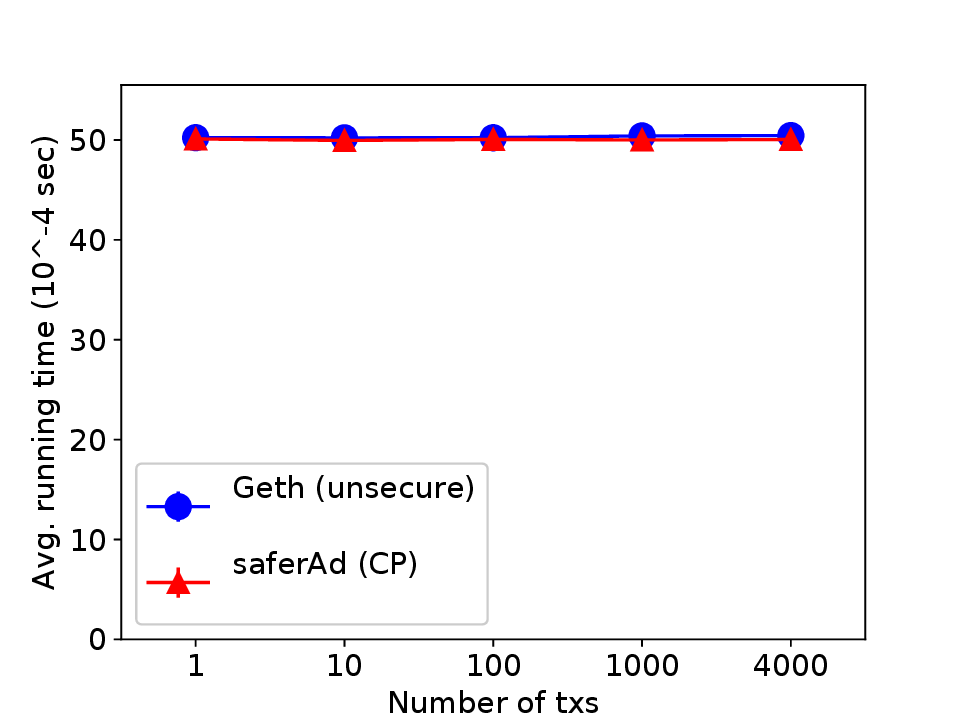}
    \label{fig:corner-td1}}%
\caption{Running time of \textsc{saferAd-CP} on Geth $v1.11.4$}%
\end{figure}

\begin{figure}[!ht]
\centering
\subfloat[Workload ``Batch insert'']{%
    \includegraphics[width=0.25\textwidth]{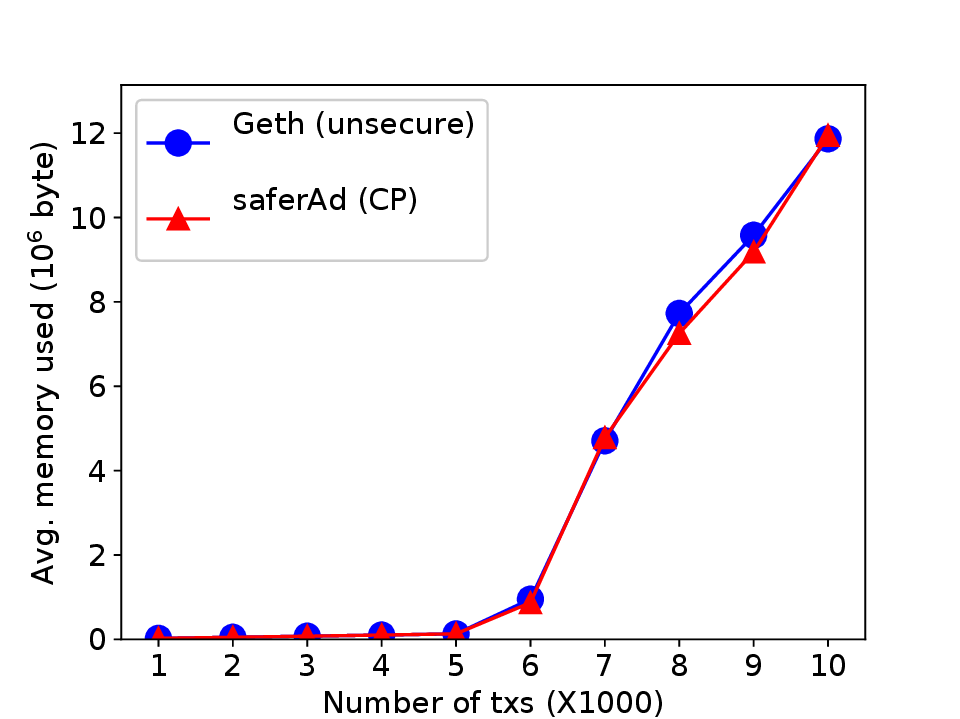}
    \label{fig:mentor-case_memory}}
\subfloat[Workload ``Mitigated attack TN1'']{%
    \includegraphics[width=0.25\textwidth]{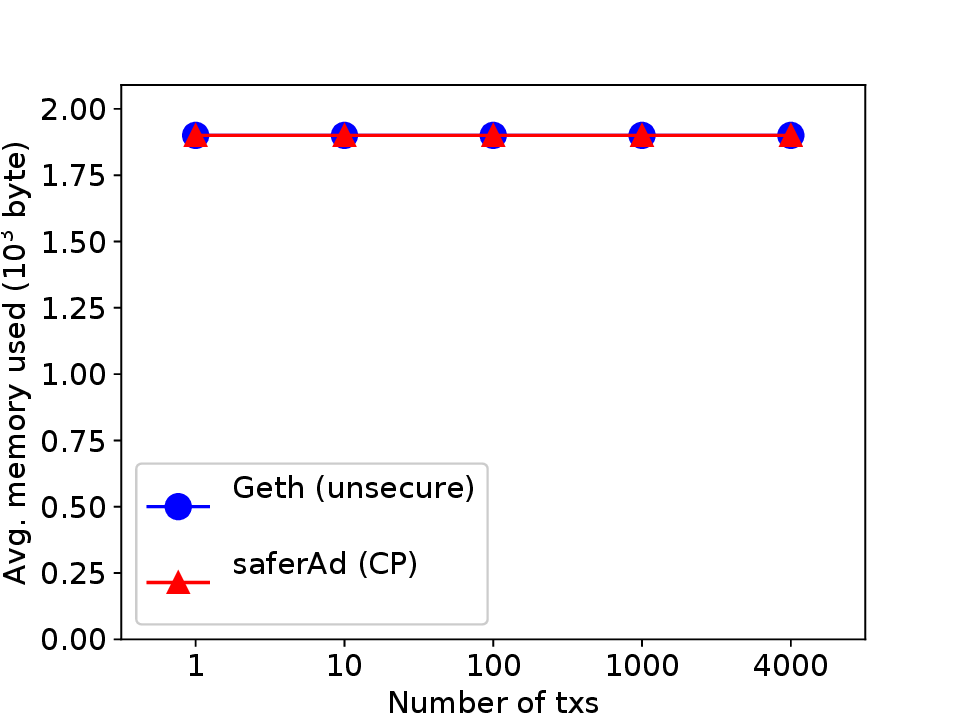}
    \label{fig:corner-td1_memory}}%
\caption{Memory usage of \textsc{saferAd-CP} on Geth $v1.11.4$}%
\end{figure}

{
We report running time and memory usage in Figure~\ref{fig:mentor-case} and Figure~\ref{fig:mentor-case_memory}, respectively. In each figure, we vary the number of transactions $n_0$ from $1000$ to $10000$, where $n_0 \in [1000, 5000]$ indicates a non-full mempool, and $n_0 \in [5000, 10000]$ indicates a full mempool. It is clear that on both figures, the performance overhead for a non-full mempool is much lower than that for a full mempool; and when the mempool is full, both the running time and memory usage linearly increase with $n_0$. The high overhead is due to the expensive mempool-reorg process triggered by a full mempool.
Specifically, 1) upon a full mempool ($n_0>5000$), on average, Geth $v1.11.4$ hardened by \textsc{saferAd-CP} causes $1.073\times$ running time and $1.000\times$ memory usage, compared to vanilla Geth $V1.11.4$.
2) When the mempool is not full ($n_0\leq{}5000$), on average, Geth $v1.11.4$ hardened by \textsc{saferAd-CP} causes $1.017\times$ running time and $1.000\times$ memory usage, compared to vanilla Geth $V1.11.4$.
Overall, the performance overhead introduced by \textsc{saferAd-CP} is negligible.
} 

Note that Geth $v1.11.4$ that mitigates DETER causes $1.083\times$ running time compared to vanilla Geth $v1.11.3$, which is vulnerable to DETER.

\noindent{\bf Workload ``Mitigated attack TN1''}: 
We additionally extend the framework with our custom workloads for countermeasure evaluation. Here, we describe the performance under one custom workload, "Mitigated attack TN1".

In the initialization phase, it first sends $n_1$ pending transactions to an empty mempool sent from $n_1'$ accounts, each of which has transactions of nonces ranging from $1$ to $\frac{n_1}{n_1'}$. All the parent transactions of nonce $1$ are of Gas price $1,000$ wei, and the others are of $200,000$ wei. It then sends $1024$ future transactions and $5120-n_1$ pending transactions from different accounts, all with nonce $1$, to fill the mempool. In the workload phase, it sends $n_1'$ transactions, each of price $20,000$ wei, to evict the parents and turn children into future transactions. 

In each experiment, we fix $n_1'$ at $10$ accounts and vary $n_1$ to be $10, 100, 1000$ and $4000$. We run the workload ten times and report the average running time and memory used with the standard deviation, as in Figure~\ref{fig:corner-td1} and Figure~\ref{fig:corner-td1_memory}. The running time and memory used are insensitive to the number of transactions ($n_1$). The performance overhead introduced by \textsc{saferAd-CP} is negligible.




\bibliographystyle{plain}
\bibliography{main}



\end{document}